\documentclass[fleqn,12pt]{article}

\usepackage{amscd}
\usepackage{amsthm}
\usepackage{amsmath}
\usepackage{amssymb}
\usepackage[applemac]{inputenc}
\usepackage[english]{babel}
\usepackage{graphicx}
\usepackage{float}
\usepackage{color}

\def\sideremark#1{\ifvmode\leavevmode\fi\vadjust{\vbox to0pt{\vss 
      \hbox to 0pt{\hskip\hsize\hskip1em           
 \vbox{\hsize2cm\tiny\raggedright\pretolerance10000
 \noindent #1\hfill}\hss}\vbox to8pt{\vfil}\vss}}}%
\newtheorem{thm}{Theorem}[section]
\newtheorem{cor}[thm]{Corollary}
\newtheorem{lem}[thm]{Lemma}
\newtheorem{prop}[thm]{Proposition}
\theoremstyle{definition}
\newtheorem{defn}[thm]{Definition}
\newtheorem{example}[thm]{Example}
\theoremstyle{remark}
\newtheorem{rem}[thm]{Remark}

\numberwithin{equation}{subsection}

\def\d{\mathrm{d} }

\begin{document}

\title{\textsc{Energy and electromagnetism of a differential $k$-form}}

\author{Navarro, J.\thanks{Department of Mathematics, University of Extremadura,
Avda. Elvas s/n, 06071, Badajoz, Spain. {\it Email address:} navarrogarmendia@unex.es \newline
The first author has been partially supported by Junta de Extremadura
and FEDER funds.}
\and  Sancho,
J.B.}


\date{\today}

\maketitle

\begin{abstract}
Let $X$ be a smooth manifold of dimension $1+n$ endowed with a Lorentzian metric $g$. The 
energy tensor of a 2-form $F$ is locally defined as $T_{ab}\, := \, -   \left( {F_a}^{i} F_{bi} - 
\frac{1}{4} \, F^{ij} F_{ij} g_{ab} \right) $.

In this paper we characterize this tensor as the only 2-covariant natural tensor associated to a Lorentzian metric and a 
2-form that is independent of the unit of scale and satisfies certain condition on its divergence. 
This characterization is motivated on physical grounds, and can be used to justify the Einstein-Maxwell field equations.

More generally, we characterize in a similar manner the energy tensor associated to a differential form of arbitrary order $k$.

Finally, we develop a generalized theory of electromagnetism where 
charged particles are not punctual, but of an arbitrary fixed dimension $p$. In this theory, the electromagnetic 
field $F$ is a differential form of order $2+p$ and its electromagnetic energy tensor is precisely the energy tensor associated 
to $F$.

\bigskip


\medskip


\end{abstract}


\bigskip

\section*{Introduction}

Let $(X,g)$ be a relativistic spacetime of dimension $1+n$; that is, $X$ is a smooth manifold of dimension $1+n$ and $g$ 
is a Lorentzian metric of signature $(+,-,\overset{n}{\dots},-)$. 

An electromagnetic field on $X$ is represented by a differential 
$2$-form $F$, and its electromagnetic energy tensor $T$ is a $2$-covariant tensor defined in a local chart by the formula:
$$ T_{ab}\, := \, -   \left( {F_a}^{i} F_{bi} - \frac{1}{4} \, F^{ij} F_{ij} g_{ab} \right) \ . $$

The main purpose of this paper is to prove the following characterization of this energy tensor:

\medskip
\noindent\textbf{Theorem:}
{\it The energy tensor is the only $2$-covariant tensor $T=T(g,F)$ naturally associated to a Lorentzian metric $g$ and a $2$-form $F$, satisfying the following properties:\smallskip

1. $T$ is independent of the unit of scale; that is, $\, T(\lambda^2g,\lambda F)\,=\, T(g,F)$ for any $\lambda >0$.\smallskip

2. At any point, $\, F_x=0 \ \Rightarrow \ T_x=0$.\smallskip

3. If $\,\d F=0\,$ then $\, \mathrm{div}\, T=-i_{\partial F}F$.}
\bigskip

More generally, the energy tensor can be defined for differential $k$-forms 
(\cite{Superenergy}), and we characterize these tensors in a similar manner (Theorem \ref{caracterizacion2}).

Let us briefly explain the role of the above theorem in order to physically motivate the definition of the energy tensor of 
an electromagnetic field. In General Relativity, 
the matter content of spacetime $X$ is represented by a symmetric $2$-covariant tensor $T_{\text{m}}$ (the matter stress-energy tensor) and, in absence of electric charges, the mass-energy and impulse conservation laws are encoded in the equation:
$$ \mathrm{div}\, T_{\text{m}} \,=\, 0 \ .$$ 

Nevertheless, when dealing with charged matter, the Lorentz force law imposes: $$\mathrm{div}\, T_{\text{m}} \,=\, i_JF 
 \,=\, i_{\partial F}F$$
where $J$ is the charge-current vector field and $\,\partial F=J^*$ because of the second Maxwell equation (in the interior product, $\partial F$ stands for the vector metrically equivalent; see our notations in Section \ref{Notations}). Therefore, in order to have
the aforementioned conservation laws, it is necessary to assume that, apart from the stress-energy tensor $T_{\text{m}}$ of the matter distribution, 
there also exists some kind of energy associated to the electromagnetic field itself, represented by some $2$-covariant tensor $T_{\text{elm}}$, such that:
$$ \mathrm{div}\, (T_{\text{m}} + T_{\text{elm}})\, =\, 0 \ .$$

Of course, this equality implies $\,\mathrm{div}\, T_{\text{elm}} = - i_{\partial F}F$. 
Since $\d F=0$ (first Maxwell equation), the tensor $T_{\text{elm}}$ has to satisfy condition ${\it 3}$ of the Theorem.

As concerns to the first hypothesis, observe that if the metric $g$ is changed by a 
proportional one, $\overline g=\lambda^2g$ ( $\lambda \in \mathbb{R}^+$), then the proper time of the trajectory of any particle 
is multiplied by the factor $\lambda$. In other words, replacing the metric $g$ by $\overline g=\lambda^2g$ amounts to a change in the time unit. This change modifies the other units of length, mass and charge, since we assume 
they are chosen in such a 
way that the universal constants (light velocity, gravitational constant and Coulomb's constant) are all equal to $1$. It is easy to check that a change of the time unit $\overline g=\lambda^2g$ (with the corresponding change in the other units) implies a modification of the type $\,\overline F=\lambda F\,$ in the mathematical representation of the electromagnetic field, while the 
matter tensor remains invariable: $\overline T_\text{m}=T_\text{m}$. 
This last equality and the equations $\,\text{div}\,(T_{\text{m}}+T_{\text{elm}})=0=
\text{div}\,(\overline T_{\text{m}}+\overline{T}_{\text{elm}})\,$, imply that the electromagnetic energy tensor also
stands invariable: $\overline{T}_{\text{elm}}=T_{\text{elm}}$. That is to say, the 
tensor $T_{\text{elm}}$ has to satisfy condition \textit{1} of the Theorem.

Finally, condition \textit{2} of the Theorem states that the electromagnetic energy is null wherever the field is null.

Summing up, these are three properties that have to be satisfied by any physically reasonable definition of electromagnetic energy tensor, and our result proves that the choice is then uniquely determined.

This problem of characterizing the electromagnetic energy tensor is classical and has already been studied in 
the literature (\cite{Anderson}, \cite{Kerrighan1}, \cite{Kerrighan2}, \cite{Lovelock}, \cite{Lovelock2}). The closest result to 
our statement is Kerrighan's (\cite{Kerrighan2}), where the tensor $T(g,F)$ is assumed to be symmetric and its coefficients are assumed to be functions of the coefficients of $g$ and $F$ (so tensors using higher derivatives of $g$ and $F$ are not considered).

Both restrictions are removed in our theorem where, instead, we require independence of the unit of scale. This is a physically meaningful condition which, in spite of its innocent appearance, turns out to be very restrictive. Some examples may illustrate this point: 

-- The Levi-Civita connection is the only linear connection naturally associated to a metric that is independent of 
the unit of scale (Epstein, \cite{Epstein}).

-- The Einstein tensor is the only, up to constant factors, $2$-covariant natural tensor associated to a metric 
that is divergence-free and independent of the unit of scale (Navarro-Sancho, \cite{Einstein}).

-- The Pontryagin forms are the only differential forms naturally associated to a metric 
that are independent of the unit of scale (Gilkey, \cite{Gilkey}).
\smallskip

Let us summarize the content of the article.

We begin with a preliminary section where we recall the definition of energy tensor associated to a 
differential $k$-form $\omega$ and the main properties that we use to characterize it.

In the following section, the problem of computing natural tensors associated to a Lorentzian metric and a $k$-form $\omega$, subject 
to a certain homogeneity condition, is reduced to a problem of representations of the orthogonal group. The main result (Theorem \ref{StredderSlovak}) 
is similar to other 
known results  (\cite{SlovakJournal}, \cite{SlovakSemi}, \cite{Stredder}), but with a slightly different homogeneity condition.

Next, we determine all the $2$-covariant natural tensors $T(g,\omega)$ that are independent of the unit of scale. As a consequence, 
it follows the announced characterization of the energy tensors (Theorem \ref{caracterizacion2}).

Finally, the existence of an energy tensor associated to a differential form of arbitrary order 
suggests the question of a possible physical interpretation for it. 
In the last section, we consider a generalized theory of electromagnetism for charged $p$-branes, introduced by Henneaux and Teitelboim (\cite{Henneaux}), where 
the electromagnetic field $F$ is a differential form of order $2+p$. 
We extend this theory up to the point of including fluids of charged $p$-branes; the corresponding Maxwell-Einstein equations require an electromagnetic energy tensor, which turns out to be  the energy tensor associated to the form $F$.\bigskip

\section{Preliminaries}\label{preliminares}

Throughout the paper, let $(X,g)$ be a Lorentzian manifold of dimension $1+n$, whose metric has signature $(+ , - , \overset{n}{\dots} , - )$. We assume $X$ is oriented, with volume form $\d X$, and time oriented.

\subsection{Notations and conventions}\label{Notations}

\smallskip

 Given a $q$-vector $D_1 \wedge \ldots \wedge D_q$ and a differential $k$-form $\omega$, with $q \leq k$, we write:
 $$ i_{D_1 \wedge \ldots \wedge D_q} \omega :=\, i_{D_q} \ldots i_{D_1} \omega  = \omega ( D_1 , \ldots , D_q , \_ \, , \ldots , \_ \, ) \ .$$

Analogously, if $\omega$ is a $k$-form and $\bar{\omega}$ is a $q$-form, with $q \leq k$, we write:
$$ i_{\bar{\omega}} \omega  :=\, i_{\bar{\omega}^* } \omega $$ where $\bar{\omega}^* $ is the $q$-vector metrically equivalent to $\bar{\omega}$.

With these notations, the metric induced on the bundle of $k$-forms is:
$$ \langle \omega , \bar{\omega} \rangle :=\, i_{\omega} \bar{\omega}\,=\, 
\frac{1}{k!}\,\omega^{j_1\dots j_k}\bar\omega_{j_1\dots j_k}  \ . $$

\smallskip
The Hodge star is the linear isomorphism $* \colon \Lambda^k T^* X \to \Lambda^{1+n-k} T^*X $ defined as:
 $$ * \omega :=\, i_{\omega} \d X \  $$
and, with these conventions, it holds: $\, * * \omega \,=\, (-1)^{(k+1)n} \omega$.

\smallskip
The codifferential $\partial \colon \Omega^k(X) \to \Omega^{k-1} (X)$ is the following differential operator:
 $$ \partial := (-1)^{(1+n)k} * \d\, * \ , \quad  \mbox{ or, equivalently,} \quad * \partial :=\, (-1)^k \d * \ . $$ 
In a local chart:
$$ (\partial \omega)_{i_1\dots i_{k-1}}\, =\, -\, \nabla^a\omega_{ai_1\dots i_{k-1}}\  . $$


\smallskip

\begin{rem}\label{FormulasAuxiliares}
 Later, we will need  the following formulae for the components of the $1$-forms $\, i_{\partial \omega} \omega\,$ and $\, i_\omega\d\omega$:
 $$\begin{aligned}
 (i_{\partial\omega}\omega)_b\,=&\, \frac{(-1)^k}{(k-1)!}\nabla_a\omega^{ai_2\dots i_k}\omega_{bi_2\dots i_k}\ ,\\
(i_\omega\d\omega)_b\,=&\,\frac{(-1)^k}{k!} \omega^{i_1\dots i_k}\nabla_b\omega_{i_1\dots i_k}+\frac{1}{(k-1)!}\omega^{i_1\dots i_k}\nabla_{i_1}\omega_{i_2\dots i_kb}\ .
\end{aligned}$$
\end{rem}

\subsection{Energy tensor of a differential $k$-form}

Let $\omega$ be a differential $k$-form on $X$.

\smallskip

\begin{defn}
Let $U$ be an observer at a point $x$ (that is, $U$ is a unitary timelike vector oriented to the future). Let us consider an orthonormal frame 
$\, ( D_0 = U, D_1 , \ldots , D_n)\,$ of $T_xX$ and the corresponding dual base $\, (\theta_0 = U^* , \theta_1 , \ldots , \theta_n)$.

In terms of this basis, the $k$-form $\omega$ decomposes as a multiple of $\theta_0$, called the \textbf{electric part} $E_U$, and
other terms without $\theta_0$, called the \textbf{magnetic part} $B_U$:
$$ \omega\, =\, E_U + B_U \,=\, (\mbox{terms with $\theta_0$} ) + ( \mbox{terms without $\theta_0$} ) \ . $$

In other words:
$$ E_U :=\, U^* \wedge i_U \omega \qquad , \qquad B_U :=\, i_U ( U^* \wedge \omega ) \ ,$$
so these $k$-forms $E_U, B_U$ depend on the observer $U$ but not on the chosen basis.

Moreover, as $E_U$ and $B_U$ are orthogonal:
$$ \langle \omega , \omega \rangle \ = \  \langle E_U , E_U \rangle + \langle B_U , B_U \rangle  \ . $$

These two addends have definite signs, that we modify to make them positive:
\begin{align*}
|E_U|^2 &:= (-1)^{k-1} \langle E_U , E_U \rangle \ = \ (-1)^{k-1} \langle i_U \omega , i_U \omega \rangle \\
|B_U|^2 &:= (-1)^k \quad  \langle B_U , B_U \rangle \ = \ (-1)^k \quad \langle U^* \wedge \omega , U^* \wedge \omega \rangle  \ .
\end{align*}
Hence,
$$  \langle \omega , \omega \rangle \  = \ (-1)^{k-1} \left( |E_U|^2 - |B_U|^2 \right) \   $$ 
and the right hand side of this equation does not depend on the observer.
\end{defn}

\smallskip

\begin{defn}
The \textbf{energy} of a differential $k$-form $\omega$ with respect to an observer $U$ is the smooth function:
$$ e(U) := \frac{1}{2} \left( |E_U|^2 + |B_U|^2 \right) \ . $$
\end{defn}

\smallskip

Unfolding the definitions, we see the energy is quadratic on $U$:
\begin{align*}
 e(U) &= \frac{1}{2} \left( |E_U|^2 + |B_U|^2 \right) = \frac{1}{2} \left( |E_U|^2 + (-1)^k \langle \omega , 
\omega \rangle + |E_U|^2 \right) \\
&= (-1)^{k-1} \left( \langle i_U \omega , i_U \omega \rangle - \frac{1}{2} \, \langle \omega , \omega \rangle \, \langle U , 
U \rangle \right) \ ,
\end{align*}
so we are led to consider the corresponding symmetric tensor:

\smallskip

\begin{defn}\label{energytensor}
The \textbf{energy tensor} of a differential $k$-form $\omega$ is the 2-covariant symmetric tensor $T$ defined as:
$$ (-1)^{k-1} \, T (D_1 , D_2) \ := \ \langle i_{D_1} \omega , i_{D_2} \omega \rangle - \frac{1}{2} \,
\langle \omega , \omega \rangle \, g (D_1 , D_2) \ .$$

This definition is made so that $T (U , U) = e(U)$ for every observer $U$. In a local chart,
$$ T_{ab}\, =\, \frac{(-1)^{k-1}}{(k-1)!}   \left( \omega_{a}^{\ i_2 \ldots i_k} \omega_{b i_2 \ldots i_k} - 
\frac{1}{2k}\, \omega^{i_1 \ldots i_k} \omega_{i_1 \ldots i_k} g_{ab} \right) \ . $$
\end{defn}

\smallskip

\begin{rem}
 These energy tensors are a particular case of the \textit{superenergy tensors} 
introduced by Senovilla (\cite{Superenergy}): the superenergy tensor associated to a differential $k$-form is 
precisely the energy tensor of Definition \ref{energytensor}.  

\end{rem}

\smallskip

Next, we quote the main property of the energy tensors, although we will not use it in this paper:

\begin{thm}[\cite{Superenergy}, Th. 4.1]
 The energy tensor $T$ of a $k$-form satisfies the dominant energy condition. 

In other words, for any pair $U_1, U_2$ of 
observers (unitary timelike vector fields oriented to the future), it holds:
$$ T (U_1 , U_2) \geq 0 \ .$$
\end{thm}

\smallskip

For any observer $U$, the Hodge star maps the electric and magnetic parts of $\omega$ into the magnetic and electric parts 
(up to signs) of $* \omega$:
$$ * E_U(\omega) \,=\, \pm B_U (*\omega) \qquad , \qquad * B_U (\omega) \,=\, \pm E_U (*\omega) \ . $$
Therefore, $\omega $ and $*\omega$ have the same energy respect to any observer $U$ and, consequently, both forms have the same energy tensor.

 The following three properties are easily obtained from the definitions, and they will suffice to characterize these energy tensors:
\begin{enumerate}
  \item Let us write $T ( g , \omega)$ to indicate that the energy tensor depends on the metric $g$ and on the $k$-form $\omega$. 
For any $\lambda > 0$, 
$$  T ( \lambda^2 g , \lambda^{k-1} \omega ) \,=\, T ( g , \omega) \ . $$

 \item At any point $x \in X$ and for any observer $U$, it holds $e(U) (x) = 0$ if and only if $\omega_x = 0$. Hence,
$$  T_x = 0 \quad \Leftrightarrow \quad \omega_x = 0 \ . $$
\end{enumerate}

\smallskip

\begin{prop}\label{divergence}
 The energy tensor $T$ of a $k$-form $\omega$ satisfies:
$$ \mathrm{div}\, T \,=\,  i_{\omega} \d \omega - i_{\partial \omega} \omega    \ .$$

In particular, if $\omega$ is closed and co-closed ($\d \omega = 0 = \partial \omega$), then $\mathrm{div}\, T = 0 $.
\end{prop}

\begin{proof} In a chart,
\begin{align*}
& (\text{div}\, T)_b\, = \, \nabla_{a} {T^a}_b  \,=\, \frac{(-1)^{k-1}}{(k-1)!} \,  \nabla_a \left( \omega^{a i_2 \ldots i_k} \omega_{b i_2 \ldots i_k} - \frac{1}{2k} \, \omega^{i_1 \ldots i_k} \omega_{i_1 \ldots i_k} \, \delta^a_b \right) \\
& =\, \frac{(-1)^{k-1}}{(k-1)!}   \left( (\nabla_a \omega^{a i_2 \ldots i_p}) \omega_{b i_2 \ldots i_k} + \omega^{a i_2 \ldots i_k} (\nabla_a \omega_{bi_2 \ldots i_k}) - \frac{1}{k} \, \omega^{i_1 \ldots i_k} \nabla_ b\omega_{i_1 \ldots i_k}  \right) \\
& \overset{\ref{FormulasAuxiliares}}{=} \, -(i_{\partial\omega}\omega)_b  + \frac{1}{(k-1)!} \, \omega^{i_1 \ldots i_k} \nabla_{i_1} \omega_{i_2 \ldots i_k b} + \frac{(-1)^k}{k!}\, \omega^{i_1 \ldots i_k}\nabla_b \omega_{i_1 \ldots i_k}  \\
&\overset{\ref{FormulasAuxiliares}}{=}\, -(i_{\partial\omega}\omega)_b +(i_\omega\d\omega)_b\quad .
 \end{align*} 
 \end{proof}
 
 \begin{cor}\label{1.10} For any $k$-form $\omega$, the $2$-covariant  tensor $\langle i_- \omega , i_- \omega \rangle$ satisfies:
 $$\mathrm{div}\,\langle i_- \omega , i_- \omega \rangle\,=\, (-1)^{k-1}\left(i_\omega\d\omega-i_{\partial\omega}\omega\right)+ \frac{1}{2}\, \d \langle \omega , \omega \rangle\ .  $$
 \end{cor}
 \begin{proof} By definition of $T$, we have:
 $$\langle i_- \omega , i_- \omega \rangle\,=\, (-1)^{k-1}T+\frac{1}{2}\langle \omega , \omega \rangle g\ .$$

Hence:
 $$
  \mathrm{div}\,\langle i_- \omega , i_- \omega \rangle\,=\, (-1)^{k-1}\text{div}\,T+\frac{1}{2}\, \text{div}\left(\langle \omega , \omega \rangle g\right) \,
 $$
 $$ =\, (-1)^{k-1}\left(i_\omega\d\omega-i_{\partial\omega}\omega\right)+\frac{1}{2}\,\d\langle \omega , \omega \rangle\ .
$$
\end{proof}

\begin{rem}
 Analogously to the case of a 2-form, the energy tensor $T$ of a $k$-form $\omega$ also appears as the Euler-Lagrange tensor 
of a variational principle. Namely, if we fix a $k$-form $\omega$ and consider the
variational problem of order 0 defined by the lagrangian density $\langle \omega , \omega \rangle \, \d X$ on the 
bundle of Lorentzian metrics, then its Euler-Lagrange
equations are precisely $T = 0$.
\end{rem}

\section{Natural tensors associated to a metric and a form}  

In this section we study natural tensors associated to a Lorentzian metric $g$ and a $k$-form $\omega$. The notion of
``natural construction'' is formalized within the language of natural bundles and natural operators. Our presentation slightly differs from the
standard approach (\cite{Kolar}), so we give a brief exposition of it.


\subsection{Natural operators}

In the following, all bundles $E \to X$ are assumed to be sub-bundles of some bundle of tensors on $X$ (observe $E \to X$ need not 
be a vector bundle; v.gr., the bundle of Lorentzian metrics). This hypothesis is not essential, but simplifies the exposition.

\smallskip

\begin{defn} Let $E,F\to X$ be two bundles over $X$ and let $\mathcal{E},\mathcal{F}$ be their sheaves of smooth sections, respectively.
  A morphism of sheaves $\,T \colon \mathcal{E} \longrightarrow \mathcal{F}\,$ is called a \textbf{regular operator} 
if, for any smooth family $\,\{ e^s \}_{s \in S}\,$ of local sections of $E$ depending on certain parameters, 
the family of sections $\,\{ T(e^s) \}_{s \in S}\,$ of $F$ also depends smoothly on those parameters.
\end{defn} 

\smallskip

According to a fundamental result due to Slovak (\cite{Peetre}), the regularity condition of a operator $T \colon \mathcal{E} \to \mathcal{F}$ implies the existence of a smooth map $\tilde T\colon J^\infty E\to F$ such that $\, T(e)=\tilde T\circ j^\infty e\,$ for all local section $e$ of $E$, so that  $T(e)$ depends on the $\infty$-jet of $e$.

\smallskip

\begin{defn} A bundle $E\to X$ (sub-bundle of a bundle of tensors) is said to be \textbf{natural} if it is 
stable with respect to the action of local diffeomorphisms of $X$.

That is, for any diffeomorphism $\tau \colon U \to V$ between open sets of $X$, it holds:
$$e\in\mathcal{E}(V)\quad\Rightarrow\quad \tau^*e\in\mathcal{E}(U)\ .$$
\end{defn}

\smallskip

\begin{defn}
Let $E,F\to X$ be natural bundles. A regular operator $T \colon \mathcal{E} \to\mathcal{F} $ is said to be 
\textbf{natural} if it is equivariant with respect to the action of local diffeomorphisms of $X$. 
 
 That is, for any diffeomorphism $\tau \colon U \to V$ between open sets of $X$, it holds:
  $$ T (\tau^*e) = \tau^* (T(e)) \ . $$
\end{defn}\medskip

\smallskip

\subsection{Normal tensors}

When dealing with jets at a point, it is useful to write the ``normal expressions" of the geometric objects under consideration.

 Let $x \in X$ be a point and let $g$ be a germ of Lorentzian metric at $x$. 
 
\smallskip

\begin{defn}
 A chart $(z_0, \ldots , z_n)$ on a neighbourhood of $x$ is said to be a \textbf{normal
 system} for $g$ at the point $x$ if the geodesics passing through $x$ at $t=0$ are precisely 
the ``straight lines" $\{ z_0(t) = \lambda_0 t , \ldots , z_n(t) = \lambda_n t \} $, where $\lambda_i \in \mathbb{R}$.
\end{defn}

\smallskip

\begin{rem}
 Via the exponential map $\mathrm{exp}_x \colon T_x X \to X$, normal systems 
at $x$ correspond bijectively to linear coordinates on $T_xX$. Therefore, two normal systems at $x$ differ on a linear transformation.
\end{rem}

Normal systems are characterized by the following well-known lemma

\begin{lem}[Gauss Lemma] 
Let $(z_0 , \ldots , z_n)$ be germs of a chart centred at $x \in X$. This chart is a normal system for a germ of a Lorentzian metric $g$ if and only if the metric coefficients $g_{ij}$ satisfy the equations:
$$ \sum_j g_{ij} z_j = \sum_{j} g_{ij}(x)\, z_j \ . $$
\end{lem}

\begin{defn} Given a a normal system $(z_0, \ldots , z_n)$ for $g$ at $x$, let us write
 $$ g_{ij, a_1 \ldots a_r} := \frac{\partial^r g_{ij}}{\partial z_{a_1} \cdots \partial z_{a_r}} (x) \ .$$
For any integer $r\geq 0$, the \textbf{$r$-th normal tensor} of $g$ at $x$  is defined to be
$$ g^r_x := \sum_{ij a_1 \ldots a_r} g_{ij,a_1\ldots a_r} \, \d_x z_i \otimes \d_x z_j \otimes \d_x z_{a_1} \otimes 
\ldots \otimes \d_x z_{a_r} \ .$$
\end{defn}
 
 Analogously, if $\omega$ is a differential $k$-form and we write:
 $$ \omega_{i_1 \ldots i_k , a_1 \ldots a_r} := \frac{\partial^s \omega_{i_1 \ldots i_k} }{\partial z_{a_1} \ldots \partial 
z_{a_s}} (x) $$ then, for any integer $s \geq 0$,
we define the \textbf{$s$-th normal tensor} of $\omega$ at $x$ to be
$$\omega^s_x := \sum_{i_1 , \ldots i_k, a_1 , \ldots a_s} \omega_{i_1 \ldots i_k, a_1 \ldots a_s} \, \d_x z_{i_1} \otimes \ldots 
\otimes \d_x z_{i_k} \otimes \d_x z_{a_1} \otimes \ldots \otimes \d_x z_{a_s} \ .
$$

\smallskip

Using the Gauss Lemma, it is easy to check that the normal tensors $g^r_x$ have the symmetries stated in the following definition.

\begin{defn}
For each integer $r \geq 1$, the vector space of \textbf{metric normal tensors} of order $r$ at $x$ is the vector subspace $N_r \subset \otimes^{r+2} T^*_x X$ of tensors $P$ with the following symmetries:
 \begin{enumerate}
 \item They are symmetric in the first two and the last $r$ indices:
 $$ P_{ij a_1 \ldots a_r} = P_{ji a_1 \ldots a_r} \quad , \quad P_{ij a_1 \ldots a_r} = P_{ij a_{\sigma(1)} \ldots a_{\sigma (r)}} \quad \forall \, \sigma \in S_r \ . $$
 
 \item The cyclic sum over the last $r+1$ indices is zero:
 $$ P_{ij  a_1 \ldots a_r} + P_{i a_r j a_1 \ldots a_{r-1} } + \ldots + P_{i a_1 \ldots a_r j } = 0 \ .$$
 \end{enumerate} 
For $r=0$, we define $N_0:=M_x$ to be the space of Lorentzian metrics at $x \in X$.
\end{defn}

\smallskip

\begin{rem}
 Due to symmetries, $N_1 = 0$, and therefore $g_x^1=0$ for any metric.
\end{rem}

\smallskip

\begin{defn}
 For any integer $s \geq 0$, the vector space of \textbf{$k$-form normal tensors} of order $s$ at $x$ is defined as:
 $$\Lambda_s :=\, \Lambda^k T^*_x X \otimes S^s T^*_x X \ . $$

Of course, for any $k$-form $\omega$ we have $\,\omega_x^s\in\Lambda_s$. 
\end{defn}

\smallskip

\begin{rem} Although we will not use this fact, let us remark that the 
sequences of tensors $\,(g_x,g^2_x,\dots, g_x^r)\,$ and $\,(g_x,R_x,(\nabla R)_x,\dots, (\nabla^{r-2}R)_x)\,$ mutually determine each other, as so happens with the sequences
 $\,(\omega_x,\omega_x^1,\dots,\omega_x^s)\,$ and  $\,(\omega_x,$ $ (\nabla\omega)_x,\dots,(\nabla^s\omega)_x)$, once the metric $g$ is fixed.
\end{rem}

\subsection{Computation of natural tensors associated to a metric and a form}

Let $\, M\subset S^2T^*X\,$ be the bundle of Lorentzian metrics on $X$, let $\Lambda^kX$ be the bundle of differential $k$-forms on $X$ and let $T_p^qX=\bigotimes^pT^*X\otimes\bigotimes^qTX$ be the bundle of $(p,q)$-tensors on $X$. 

Their sheaves of smooth sections will be written, respectively,
$$Metrics\qquad,\qquad Forms_k\qquad,\qquad Tensors_{p}^q\ .$$

\begin{defn}
A \textbf{natural $(p,q)$-tensor}, associated to a Lorentzian metric and a differential $k$-form, is a natural operator $\, T \colon Metrics \times Forms_k \longrightarrow Tensors_{p}^q$.
\end{defn}

\begin{defn} A natural $(p,q)$-tensor $T$ is said to be \textbf{homogeneous of weight} $w\in \mathbb{R}$ if, for any metric $g$, any $k$-form $\omega$ and any real number $\lambda > 0$, it holds:
 $$ T ( \lambda^2 g , \lambda^{k-1} \omega ) = \lambda^w\, T ( g, \omega)  \ . $$
 
 If $T$ is homogeneous of weight 0, we say it is \textbf{independent of the unit of scale}.  
 \end{defn}

\smallskip

Let $x \in X$ be a point and let $g_x$ be a Lorentzian metric at $x$. We will write $O_{g_x} := O(1,n)$ for the orthogonal group of $(T_x X , g_x)$. 
The symmetric powers $S^d N_r$ and $S^{c} \Lambda_s$ are linear representations of $O_{g_x}$.

If $V$ and $W$ are linear representations of $O_{g_x}$, 
we denote $\mathrm{Hom}_{O_{g_x}} (V , W)$ the vector space of $O_{g_x}$-equivariant linear maps $V \to W$.

 The following theorem allows to compute homogeneous natural tensors:

\begin{thm}\label{StredderSlovak} 
Let us fix a point $x\in X$ and a Lorentzian metric $g_x$ at $x$.
There exists an $\mathbb{R}$-linear isomorphism:
$$\begin{CD}
  \{ \text{Natural (p,q)-tensors homogenous of weight }w \,  \}  \\ 
  @| \\ 
     \bigoplus\limits_{\{ d_i , c_j \} } \mathrm{Hom}_{O_{g_x}} (S^{d_2}N_2 \otimes \cdots \otimes
S^{d_r}N_r \otimes S^{c_0} \Lambda_0 \otimes \ldots \otimes S^{c_s} \Lambda_s \ ,\  (T_{p}^q X)_x )
\end{CD}$$
where the summation is over all sequences of non-negative integers $\{d_2 , \ldots , d_r \} $, $r\geq 2$, and $\{ c_0 , \ldots , c_s \}$, satisfying
the equation:
\begin{equation}\label{Condicion}
2d_2 + \ldots + r\, d_r + c_0 + 2 c_1 + \ldots + (s+1) c_s = p-q -w\ .
\end{equation}
If this equation has no solutions, the above vector space is reduced to zero.
\end{thm}

\smallskip

This theorem is closely related to results of Stredder (\cite{Stredder}, Theorem 2.5) and Slovak (\cite{SlovakJournal}, Theorem 3.3) 
and the proof is similar. 
\smallskip

\begin{rem} If $\phi \colon S^{d_2}N_2 \otimes \cdots \otimes S^{d_r} N_r \otimes S^{c_0} \Lambda_0 \otimes \ldots \otimes S^{c_s} 
\Lambda_s \to (T_{p}^q X)_x$ is an $O_{g_x}$-equivariant linear map, then the corresponding natural tensor
$T(g,\omega)$ is obtained by the formula:
$$ T(g , \omega)_{x} = \phi \left( (g^2_{x} \otimes \stackrel{d_2}{\ldots} \otimes \, g^2_{x}) \otimes \cdots
\otimes (g^r_{x} \otimes \stackrel{d_r}{\ldots} \otimes \, g^r_{x}) 
 \otimes \ldots 
\otimes ( \omega^s_x \, \otimes \stackrel{c_s}{\ldots} \otimes \, \omega^s_x)  \right) $$ where $(g^2_x, g^3_x, \ldots )$ is
the sequence of metric normal tensors of $g$ at $x$ and $( \omega^0_x , \omega^1_x , \ldots )$ is the sequence of $k$-form normal tensors of $\omega$  at $x$. In this equality, $g$ is assumed to have the prefixed value at $x$. 
\end{rem}

\smallskip

\smallskip

\begin{rem}\label{CalculoAplicaciones}
 The $O_{g_x}$-equivariant linear maps that appear in the theorem can be explicitly computed using the
isomorphism:
$$ 
\begin{array}{c}
\mathrm{Hom}_{O_{g_x}} \left( S^{d_2}N_2 \otimes \cdots \otimes S^{c_s} \Lambda_s \ ,\
(T_{p}^qX)_x \right) \\
\parallel \\
\mathrm{Hom}_{O_{g_x}} \left( S^{d_2} N_2 \otimes \cdots \otimes S^{c_s} \Lambda_s \otimes
(T_{q}^pX)_x \ ,\ \mathbb{R} \right) \ .
\end{array}
$$ 
and
applying the Main Theorem of the invariant theory for the orthogonal group $O_{g_x}$ (see \cite{Jaime}, Th. 4.1, for a proof in
the Lorentzian case).
This theorem states that any $O_{g_x}$-equivariant linear map $S^{d_2} N_2 \otimes \cdots \otimes (T_{q}^pX)_x
\to \mathbb{R}$ is a linear combination of iterated contractions with respect to the metric $g_{x}$.\goodbreak

As an example, for a non zero linear map to exist, the total order (covariant plus contravariant order) of the space of tensors 
$S^{d_2} N_2 \otimes \cdots \otimes (T_{q}^pX)_x $ has to be even. 
\end{rem}

\smallskip

\section{Characterization of the energy tensors}

In this section, we characterize the energy tensor of a $k$-form by three conditions (Theorem \ref{caracterizacion2}). To do so, we analyse separately the consequences of each of these conditions.

\begin{prop} 
Let $\, T\colon Metrics\times Forms_k\longrightarrow Tensors_2^0\,$ be a natural tensor, with $k\neq 1,3$. 

If it is independent of the unit of scale, then $T(g, \omega)$ is an $\mathbb{R}$-linear combination of the following four tensors:
$$Ricci(g)\quad ,\qquad r(g)g\quad ,\qquad \langle i_{-} \omega , i_{-} \omega \rangle_g\quad ,\qquad \langle \omega , \omega \rangle_g\, g\quad $$ where $r(g)$ stands for the scalar curvature of the metric $g$.
\end{prop}

\begin{proof} By Theorem \ref{StredderSlovak}, such an homogeneous natural tensor of weight $w=0$ is determined by a $O_{g_x}$-equivariant linear map:
$$\begin{CD}
 S^{d_2} N_2 \otimes \ldots \otimes S^{d_r} N_r \otimes S^{c_0} \Lambda_0 \otimes \ldots \otimes 
S^{c_s} \Lambda_s @>>> \bigotimes^2T^*_xX
\end{CD} $$ 
where the coefficients $d_i,c_j \in \mathbb{N}$ satisfy the equation:

$$ 2d_2 + \ldots + rd_r + c_0 + 2 c_1 + \ldots + (s+1) c_s = 2 \ .$$

If some $c_i$ is non zero, then there are only two possibilities:

\medskip
$\bullet$ $c_0 = 2$, $c_1 = \ldots = c_s = d_j = 0$. In this case, we are reduced to compute $O_{g_x}$-equivariant linear maps:

$$ S^2 ( \Lambda^k_xX ) \otimes T_xX \otimes T_xX \longrightarrow \mathbb{R} \ .$$

As explained in Remark \ref{CalculoAplicaciones}, those linear maps are linear combinations of iterated contractions. Due to symmetries, any such an 
iterated contraction is a linear combination of these two:
$$
T \longmapsto T_{a_1\dots a_ka_1\dots a_kbb}\qquad,\qquad
T \longmapsto T_{ba_2\dots a_kca_2\dots a_kbc} \ .$$
These contractions, in turn, correspond, respectively, with the tensors:
$$ \langle \omega , \omega \rangle_g g \quad , \quad \langle i_{-} \omega , i_{-} \omega \rangle _g \ .$$

\medskip

$\bullet$ $c_1 = 1$, $c_0 = c_2 \ldots = c_s= d_j = 0$: Any such a tensor would produce a $O_{g_x}$-equivariant linear map:
$$ \Lambda^k_xX \otimes T^*_xX \otimes T_xX \otimes T_xX \longrightarrow \mathbb{R} \ .$$
but there are no such maps for $k$ even or $k \geq 5$ (the contraction of two skew-symmetric indices is zero).

\medskip

Finally, if the $c_i$ are all zero, then the tensor $T(g,\omega)$ does not depend on $\omega$ and therefore it is a linear combination of $\, Ricci(g)\,$ and $\, r(g)g\,$ (see details on \cite{Einstein}, Theorem 5.1).

\end{proof}

\smallskip

\begin{rem} In the previous Proposition, if $k = 3$ there also exists 
the $O_{g_x}$-invariant linear map:
$$ \Lambda^3 T^*_xX \otimes T^*_xX \otimes T_xX \otimes T_xX \to \mathbb{R} \quad , \quad T\longmapsto T_{abcabc}$$
 that corresponds to the natural tensor $C_{01} (\nabla \omega)$, where $C_{01}$ denotes the contraction of the first two indices.

If $k=1$,
$$ T^*_xX \otimes T^*_xX \otimes T_xX \otimes T_xX \to \mathbb{R} $$ there exist three different iterated contractions:
$$T\longmapsto T_{aabb}\qquad,\qquad T\longmapsto T_{abab}\qquad ,\qquad T\longmapsto T_{abba}\ ,$$
 that correspond, respectively, with the natural tensors:
$$ (\mathrm{div}_g\, \omega ) g \quad , \quad \nabla_g\, \omega \quad , \quad (\nabla_g\, \omega)^t $$ where 
$(\nabla_g\, \omega)^t (D_1 , D_2) := (\nabla_g\, \omega) (D_2 , D_1)$.\medskip
\end{rem}

\smallskip

\begin{prop} Let $\, T\colon Metrics\times Forms_k\longrightarrow Tensors_2^0\,$ be a natural tensor. If it satisfies:
\begin{itemize}
\item[1)] $T$ is independent of the unit of scale: $\, T (\lambda^2 g , \lambda^{k-1} \omega) = T( g , \omega)\,$ for all $\lambda>0$,

\item[2)] At any point, $\quad\omega_x = 0 \ \Rightarrow \ T(g,\omega)_x = 0$,
\end{itemize}
then there exist universal constants $\mu_1,\mu_2 \in \mathbb{R}$ such that:
$$ T(g,\omega)\,=\,\mu_1\,\langle i_{-} \omega , i_{-} \omega \rangle_g +\mu_2\,\langle \omega , \omega \rangle_g\, g\ .$$

\end{prop}

\begin{proof} Condition (2) rules out the tensors $\, Ricci(g)\,$ and $\, r(g)g\,$ in the previous proposition, as well as the other exceptional tensors in the cases $k=1,3$.

\end{proof}

\smallskip

\begin{thm}\label{caracterizacion} If a natural tensor $\, T\colon Metrics\times Forms_k\longrightarrow Tensors_2^0\,$ satisfies:
\begin{enumerate}
\item It is independent of the unit of scale: $\, T (\lambda^2 g , \lambda^{k-1} \omega) = T ( g , \omega)\,$ for all $\lambda>0$,

\item At any point, $\quad\omega_x = 0 \ \Rightarrow \ T(g,\omega)_x = 0$,

\item $\mathrm{div}_g \, T(g,\omega) = 0$ whenever $\omega$ is closed and co-closed, 
\end{enumerate}
then $\, T(g,\omega)\,$ is a constant multiple of the energy tensor:
$$E(g,\omega) :=(-1)^{k-1}\left( \langle i_{-} \omega , i_{-} \omega \rangle_g - \frac{1}{2}  \langle \omega , \omega \rangle_g \, g \right) .$$
\end{thm}

\begin{proof} By the previous proposition, there exist universal constants $\mu_1,\mu_2 \in \mathbb{R}$ such that:
$$T(g,\omega)\,=\,\mu_1\langle i_- \omega , i_- \omega \rangle+\mu_2\langle \omega , \omega \rangle \, g\ .$$

Then, writing $T=T(g,\omega)$,
\begin{align*}
\text{div}\, T\, &=\, \ \mu_1 \, \text{div}\, (\langle i_- \omega , i_- \omega \rangle)+\mu_2\, \text{div}\, (\langle \omega , \omega \rangle \, g) \\ 
&\overset{\ref{1.10}}{=}\,\mu_1\left( (-1)^{k-1}\left(i_\omega\d\omega-i_{\partial\omega}\omega\right)+\frac{1}{2}\d\langle \omega , \omega \rangle\right)+\mu_2\d\langle \omega , \omega\rangle\ 
\end{align*}

By hypothesis (3), if $\omega$ is closed and co-closed, then $\text{div}\, T=0$. Comparing with the previous equation, we have
$$0\,=\, \text{div}\, T\,=\, \mu_1\left( 0+\frac{1}{2}\, \d\langle \omega , \omega \rangle\right)+\mu_2\, \d\langle \omega , \omega\rangle\ ,$$
hence $\mu_2=-\mu_1/2$ and we conclude:
$$T\,=\,\mu_1\langle i_- \omega , i_- \omega \rangle-\frac{\mu_1}{2}\langle \omega , \omega \rangle \, g\,=\, \mu_1(-1)^{k-1}E\ .$$
\end{proof}

\smallskip 

We may reformulate the above theorem so as to eliminate the constant factor:

\smallskip

\begin{thm}\label{caracterizacion2} If a natural tensor $\, T\colon Metrics\times Forms_k\longrightarrow Tensors_2^0\,$ satisfies:
\begin{enumerate}
\item It is independent of the unit of scale: $T (\lambda^2 g , \lambda^{k-1} \omega) = T ( g , \omega)$ for all $\lambda>0$,

\item At any point, $\quad\omega_x = 0 \ \Rightarrow \ T(g,\omega)_x = 0$,

\item $\mathrm{div}_g \, T(g,\omega) = -i_{\partial\omega}\omega$ whenever $\omega$ is closed,
\end{enumerate}
then $\, T(g,\omega)\,$ coincides with the energy tensor $\,E(g,\omega)$.
\end{thm}

\begin{proof} It is clear that $T$ satisfies the hypotheses of the previous theorem, so  $T$ is a constant multiple of the energy tensor $E$. As both tensors have the same divergence whenever $\omega$ is closed (see Proposition \ref{divergence}), that constant has to be one.

\end{proof}

\smallskip

\begin{rem}
Let $ClosedForms_k$ be the sheaf of closed $k$-forms on $X$. We state, without proof, another variation of the previous result:

\medskip

\noindent {\bf Theorem.} {\it 
Let  $\, T:Metrics\times ClosedForms_k\longrightarrow Tensors_{2}^0\,$ be a natural operator. If it satisfies:
\begin{enumerate}
\item It is independent of the unit of scale: $T (\lambda^2 g , \lambda^{k-1} \omega) = T ( g , \omega)$ for all $\lambda>0$.

\item At any point, $\quad\omega_x = 0 \ \Rightarrow \ T(g,\omega)_x = 0$.

\item $\mathrm{div}_g \, T(g,\omega) = -i_{\partial\omega}\omega$.
\end{enumerate}
then $\, T(g,\omega)\,$ coincides with the energy tensor $\,E(g,\omega)$.}
\end{rem}

\section{Electromagnetism of $p$-branes}

There exists  a theory of electromagnetism for charged $p$-branes (\cite{Henneaux}), where the electromagnetic field  is represented by a differential $(p+2)$-form $F$. In the rest of the paper, we shall extend this theory up to the point of including a force law for fluids of charged $p$-branes and an electromagnetic energy tensor, necessary to state the Einstein equation. This tensor is precisely the energy tensor of the form $F$ introduced in Definition \ref{energytensor}.

In this section we analyse the interaction of a charged $p$-brane with an arbitrary electromagnetic field. Our analysis is developed at the classical (non quantum) level and, in contrast to  \cite{Henneaux}, it is based on the elementary concepts of impulse and  acceleration of a $p$-brane.

\smallskip

From now on, let us fix an integer $p$, such that $0 \leq p \leq n$, and let us write $q := n-p$.


\begin{defn}
 The \textbf{trajectory} 
\textbf{of a $p$-brane} is, by definition, an oriented smooth submanifold $S \subset X$ of 
 dimension $p+1$, whose metric $g_{|S}$ has signature $(+, - , \stackrel{p}{\ldots} , - )$.

 Associated to any $p$-brane, we also assume two constants, called \textbf{tension} $\mathfrak{t}>0$ 
(or \textit{mass}, in the case $p=0$ of punctual particles), and \textbf{electric charge} $\mathfrak{q} \in \mathbb{R}$.
\end{defn}

\subsection{Impulse form of a $p$-brane}

In absence of external forces, the trajectory of a punctual particle is a geodesic of spacetime. To extend this fundamental principle to the movement of a $p$-brane, let us recall two different characterizations of geodesics: 

\begin{enumerate}
\item The trajectory of a particle is a geodesic if the impulse vector $mU$ is parallel along the trajectory (where $m$ is the mass of the particle and $U$ is the unitary tangent vector to the trajectory).

\item The trajectory of a particle is a geodesic if it minimizes the action $m\int\text{d}\tau$, where $\tau$ is the proper time of the trajectory.
\end{enumerate}

To determine the movement of a $p$-brane in absence of external forces, it is common in the literature to follow the second approach, using variational principles. To be precise, the generalized action is the Nambu-Goto action, $\,\mathfrak{t}\int_S\text{d}S$, where $\text{d}S$ is the  $(p+1)$-volume of the trajectory $S$ of the brane.  

Instead of that, in this paper we generalize the concept of impulse to a $p$-brane, arriving to the same equations of motion.\medskip

Let $S \subset X$ be the trajectory of a $p$-brane and let $\d S$ be the $(p+1)$-volume form of $S$. Rising the first index of $\d S$ and multiplying it by the tension $\mathfrak{t}$, we obtain 
a $p$-form with values on tangent vectors, that is called the 
\textbf{impulse form} of $S$. In other words,

\smallskip

\begin{defn}
 The \textbf{impulse form} of a $p$-brane $S$ is the $p$-form on $S$ with values on $TS$:
$$\begin{CD}
 \Pi_S \ \colon \ TS \wedge \stackrel{p}{\ldots} \wedge TS @>>>
 TS \subset (TX)_{|S}
 \end{CD}$$
  defined by the following property:
$$ g(D_0, \Pi_S (D_1 , \ldots , D_p))\, =\, \mathfrak{t} \, \d S (D_0 , \ldots , D_p) $$ for any $D_0 , \ldots , D_p$ tangent vectors to $S$.\goodbreak

If $D_0 , \ldots , D_p$ is an orthonormal frame of vector fields on $S$, where the matrix of $g_{|S}$ is diagonal (+1, -1, \ldots , -1), 
then:
$$ \Pi_S \,=\, \mathfrak{t}\, \sum_{j=0}^p (i_{D_j} \d S) \otimes \delta_j D_j   $$ where $\delta_0 = 1$ and $\delta_{j} = -1$ for $j \neq 0$.
\end{defn}

\begin{example}
In the case of a particle $(p=0)$, the trajectory $S$ is a curve and the impulse form $\Pi_S$ is a 
vector-valued 0-form; that is, it is simply a tangent vector $\Pi_S = m U$, where $m$ is the mass of the particle and $U$ is the unitary tangent vector to the curve.
\end{example}

\begin{example}\label{PrimerMinkowski}
 If $(X = \mathbb{R}^{1+n} , g = \d t^2 - \sum_i \d x_i^2 )$ is the Minkowski spacetime, then the impulse form of a $p$-brane $\,S\,$ can be written as:
$$ \Pi_S \,=\, \omega_0 \otimes \partial_t + \omega_1 \otimes \partial_{x_1} + \ldots + \omega_n \otimes \partial_{x_n} $$ for some ordinary 
differential $p$-forms $\omega_i$ on $S$.\goodbreak 

If $S_{t_0} := S \cap \{ t = t_0 \} $ is the particle at the instant $t_0$, then the vector:
$$ \int_{S_{t_0}} \Pi_S := \left( \int_{S_{t_0}} \omega_0 \right) \partial_t + \ldots + \left( \int_{S_{t_0}} \omega_n \right) \partial_{x_n} $$
can be understood as the total energy-impulse vector of the $p$-brane at $t_0$.

The differential $p$-form $\omega_0$ is called \textbf{energy form} of the brane respect to the chosen 
inertial frame $(t, x_1, \ldots , x_n)$, and the integral:
$$ \int_{S_{t_0}} \omega_0 $$ is understood as the total energy of the brane at $t_0$.

Moreover, if the $p$-brane is at apparent rest at $t_0$ (that is, $\partial_t$ is tangent to $S$ at the points $t = t_0$) then it is easy to check  the total energy of the brane at $t_0$ is equal to:
$$ \int_{S_{t_0}} \omega_0 \,=\, \mathfrak{t} \cdot ( \mbox{Volume of } S_{t_0} )  \ .$$
\end{example}

\smallskip


\begin{defn}
Let $S \subset X$ be the trajectory of a $p$-brane, and let us write $\nabla $ for the Levi-Civita connection of $(X,g)$.
For any pair of tangent vector fields $D,D'$ on $S$, the covariant derivative $\nabla_D D'$ decomposes as a tangent vector to $S$ plus a vector orthogonal to $S$:
$$ \nabla_D D' \,=\, \text{tang}( \nabla_D D') + \text{nor} ( \nabla_D D') \ .$$

The first addend $\overline{\nabla}_D D' := \text{tang}(\nabla_D D') $ is precisely the covariant derivative with respect to the Levi-Civita connection
$\overline{\nabla}$ of the submanifold $(S,g_{|S})$.

The second addend is, by definition, the \textbf{second fundamental form} of $S$, which is a symmetric tensor with values on 
the normal bundle of $S$:
$$ \Phi_S \colon TS \times TS \longrightarrow (TS)^{\perp} \quad , \quad \Phi_S (D, D') :=\,\text{nor}(\nabla_D D') \ .  $$

Therefore, the \textbf{trace} of the second fundamental form, $\mathrm{tr}\, \Phi_S$, is a field of normal vectors to $S$.
\end{defn}

\smallskip

\begin{prop}\label{DiferencialImpulso}
 The impulse form $\Pi_S$ of a $p$-brane $S$ satisfies: 
$$ \d_\nabla \Pi_S \,=\, \d S \, \otimes \, \mathfrak{t}\cdot \mathrm{tr} (\Phi_S) \ .$$
\end{prop}

\begin{proof}
 Let $(D_0, \ldots , D_n)$ be an orthonormal basis of vector fields on $X$, such that $(D_0 , \ldots , D_p)$ is an orthonormal basis of vector fields on $S$.
 
 Let us write
 $$ \d_\nabla D_j \,=\, \sum_{i=0}^n \omega_{ij} \otimes D_i $$ where the $\omega_{ij}$ are the connection 1-forms.
 
 Consequently,
\begin{align*}
 \nabla_{D_j} D_j &=  \sum_{i=0}^n \omega_{ij} (D_j) \, D_i \ , \quad  \mbox{ and } \quad 
 \overline{\nabla}_{D_j} D_j = \sum_{i=0}^p \omega_{ij} (D_j) \, D_i \ , \quad j \leq p \ , \\
 \mathrm{div}_S D_j &= \textit{contr}\,( \d_{\overline{\nabla}} D_j ) = \textit{contr} \left( \sum_{i=0}^p \omega_{ij} \otimes D_i \right) = \sum_{i=0}^p \omega_{ij} (D_i) \ ,
\end{align*} where $contr$ denotes the contraction of the contravariant and covariant indexes.

Using these formulae, 
and taking $\mathfrak{t}=1$:
\begin{align*}
 \d_\nabla \Pi_S &=   \sum_{j=0}^p \d ( i_{D_j} \d S) \otimes \delta_j D_j + (-1)^p \sum_{j=0}^p (i_{D_j} \d S) \wedge \delta_j \d_\nabla D_j \\
 &= \sum_{j=0}^p (\mathrm{div}_S D_j) \, \d S \otimes \delta_j D_j + (-1)^p \sum_{j=0}^p \sum_{i=0}^n (i_{D_j} \d S) \wedge \delta_j \omega_{ij} \otimes D_i \\
 &= \d S \otimes \sum_{j=0}^p \delta_j ( \mathrm{div}_S D_j) D_j + \sum_{j=0}^p \sum_{i=0}^n \delta_j \omega_{ij} (D_j)) \d S \otimes D_i \\
 \end{align*}
(applying the formulae for $\mathrm{div}_S D_j$ and $\nabla_{D_j}D_j$)
\begin{align*}
&= \d S \otimes \sum_{j=0}^p \sum_{i=0}^p \delta_j \omega_{ij}(D_i)D_j  + \sum_{j=0}^p \d S \otimes \delta_j \nabla_{D_j} D_j \\
& = - \d S \otimes \sum_{j=0}^p \sum_{i=0}^p \delta_i \omega_{ji} (D_i) D_j + \d S \otimes \sum_{j=0}^p \delta_j 
\nabla_{D_j} D_j \\
& = - \d S \otimes \sum_{i=0}^p \delta_i \overline{\nabla}_{D_i} D_i + \d S \otimes \sum_{j=0}^p \delta_j \nabla_{D_j} D_j \\
& = \d S \otimes \sum_{j=0}^p \delta_j \left( \nabla_{D_j} D_j - \overline{\nabla}_{D_j} D_j \right) = \d S \otimes \sum_{j=0}^p \delta_j \Phi_S (D_j , D_j) = \d S \otimes \mathrm{tr}\, \Phi_S \ . 
\end{align*}

\end{proof}\goodbreak

\smallskip

\begin{example}
In case $p=0$, let $S$ be the trajectory  of a particle with impulse $\Pi_S = m U$, where $m$ is the mass of the particle and $U$ is the future-pointing unitary tangent vector of the curve $S$. 

Since:
$$\Phi_S(U,U)=\text{nor}(\nabla_UU)=\nabla_UU=\nabla_{\partial_\tau}U$$
we observe $\mathrm{tr}\,\Phi_S=\nabla_{\partial_\tau}U$ is the acceleration vector of the particle.
\end{example}

\begin{defn}
By analogy with the particle case just explained, if $S$ is the trajectory of a 
$p$-brane, then the normal vector $\mathrm{tr}\, \Phi_S$ is interpreted as the \textbf{acceleration} of the brane.
\end{defn}

If there are no external forces, the trajectory $S$ of a brane should have null acceleration. For a particle, this amounts to 
saying that it is a geodesic: $\nabla_{\partial_{\tau}} U = 0$. For a $p$-brane, this amounts to the equation:
$$ \mbox{ {\bf Inertial Motion:} } \quad \boxed{\phantom{\frac{1}{2}} \mathrm{tr}\, \Phi_S = 0 \  \, } \ .$$  

By \ref{DiferencialImpulso}, this equation is equivalent to
$\,\d_\nabla \Pi_S = 0$, which is an infinitesimal conservation law for the impulse.

\smallskip

\begin{rem}
The equation $\mathrm{tr}\, \Phi_S = 0$ is precisely the Euler-Lagrange equation for the variational problem defined by the Nambu-Goto action.
\end{rem}

\subsection{Electromagnetic field}

\smallskip

\begin{defn}\label{4.10}
 An \textbf{electromagnetic field} over the spacetime $X$ is a
skew-symmetric tensor:
$$\widehat F \, \colon \, TX \wedge \stackrel{p+1}{\ldots\ldots} \wedge TX \longrightarrow TX $$ satisfying the following property:
$$ \widehat F (D_0 , \ldots , D_p) \, \in \ < D_0 , \ldots , D_p >^{\perp} $$ for any collection $D_0, \ldots , D_p$ of vector fields on $X$.\smallskip

The value $\widehat F (D_0 , \ldots , D_p)_x$ may be understood as the force at the point $x$ that suffers a brane with $(p+1)$-volume vector 
$D_0 \wedge \ldots \wedge D_p$ and unitary charge (see the force law below).

The definition of $\widehat F$ amounts to saying that the tensor:
$$F (D_0 , \ldots , D_{p+1}) := \, g(\widehat F (D_0 , \ldots , D_p) , D_{p+1}) $$ is a $(p+2)$-differential form on $X$, and we will say 
that $F$ is the \textbf{$(p+2)$-form of the electromagnetic field}.
\end{defn}

\subsubsection*{Force Law for a $p$-brane} 

Let $S$ be the trajectory of a $p$-brane with tension $\mathfrak{t}$ and electric charge $\mathfrak{q}$.

\begin{defn}
 The \textbf{charge-current vector} of this brane is the only $(p+1)$-vector $J_S$ on $S$ satisfying $$\d S (J_S) \,=\, \mathfrak{q} \ . $$

If $(D_0, \ldots , D_p)$ is an oriented orthonormal frame of vector fields on $S$, then:
$$ J_S \,=\, \mathfrak{q}\,  D_0 \wedge \ldots \wedge D_p \ . $$
\end{defn}

\smallskip

Let $\widehat{F}$ be an electromagnetic force and assume that the $p$-brane $S$ does not substantially modify the electromagnetic field. Nevertheless, the $p$-brane $S$ does suffer an acceleration due to the electromagnetic force $\widehat F$, that we postulate to be governed by the following equation:

$$ \mbox{\textbf{Lorentz Force Law}:} \quad \boxed{ \ \d_\nabla \Pi_S \, =\, \d S \otimes \widehat F(J_S) \ } \ . $$

Using Proposition \ref{DiferencialImpulso}, this equation is equivalent to $\mathfrak{t} \cdot \mathrm{tr}\, \Phi \,=\, \widehat F (J_S)$, which, substituting the value of $J_S$, is in turn equivalent to: 

$$\mathfrak{t} \cdot \mathrm{tr}\, \Phi \,=\, \mathfrak{q} \cdot \, \widehat F (D_0, \ldots , D_p)\  .$$

\smallskip

Observe the typical form of this equation: 
$ mass \, \times \,  acceleration = force$; the definition \ref{4.10} of $\widehat F$ has been dictated by the need of giving sense to this expression.


\smallskip

\begin{example}
In the case $p=0$, the charge-current vector of a particle is simply a vector $J_S = \mathfrak{q}\, U$, 
where $U$ is the future-pointing unitary  tangent vector of the  trajectory $S$ of the particle.

Since the impulse of the particle is $mU$, the force law reads:
$$ \d_\nabla ( mU ) \,=\, \d \tau \otimes\mathfrak{q} \, \widehat F (U) $$ 
where $\tau$ stands for the proper time of the curve. 

As $\d_\nabla U = \d \tau \otimes \nabla_{\partial_\tau} U$, this force law is equivalent to the equation:
$$ m \cdot \nabla_{\partial_{\tau}} U \,=\, \mathfrak{q} \cdot \widehat F (U) $$ which is precisely the classical
\textbf{Lorentz Force Law} for a particle of mass $m$ and charge $\mathfrak{q}$.
\end{example}

\smallskip

\begin{rem}

Let $(X=\mathbb{R}^{n+1}, g=\d t^2-\sum^n_1\d x_i^2)$ be the Minkowski space-time. The trajectory of a $p$-brane $S$ can be written as $x_i = f_i (t, u_1, \ldots , u_p)$, where $(t, u_1, \ldots , u_p)$ are local coordinates on $S$.
 
 On these coordinates $(t,u_1, \ldots , u_p)$, the force law $\mathfrak{t}\cdot \mathrm{tr}\, \Phi = \widehat{F} (J_S) $ produces a system of second order partial differential equations, that is quasi-linear and hyperbolic. 
 
 For these kind of systems, the Cauchy problem has a unique local solution (\cite{Taylor}, Proposition 3.2), so the force law uniquely determines the trajectory of the $p$-brane, for adequate initial conditions. 
\end{rem}

\subsection{Maxwell equations}

\begin{defn}
A distribution of charged $p$-branes on the spacetime $X$ will be represented by means of a differential $q$-form $C$, called the \textbf{charge-density form}.
\end{defn}

The physical meaning of this $q$-form is the following (recall $q=n-p$): Given $q$ linearly independent vectors $D_1, \ldots , D_q \in T_xX$, that we understand as an oriented infinitesimal parallelepiped at the point $x$, we have:
$$ C(D_1 , \ldots , D_q) = \left\{ 
\begin{array}{c}
 \mbox{Sum, affected with a sign, of the charges of the p--branes } \\
 \mbox{ transversally crossing the parallelepiped }
\end{array} \right\}  \ . $$

We say that a $p$-brane with a trajectory $S$ transversally crosses the parallelepiped $D_1, \ldots , D_q$ whenever $T_x X = T_x S \oplus \langle D_1, \ldots , D_q \rangle$. 
If the orientation of $T_xX$ coincides with the product of the orientations on $T_xS $ and $\langle D_1 , \ldots , D_q \rangle$, then the charge of the $p$-brane counts with positive sign; otherwise, we affect the charge with a negative sign.

\smallskip

\begin{defn}
 The \textbf{charge-current $(p+1)$-vector} of a distribution of charged $p$-branes is the only $(p+1)$-vector $J$ satisfying:
 $$ i_J \d X \,=\, C \ .$$

 Equivalently, if $J^*$ is the $(p+1)$-form metrically equivalent to $J$ and $*$ stands for the Hodge operator,
 $$ J^* \,=\, (-1)^{pn} \, *C \ .$$
\end{defn}

\begin{example}
When $p=0$, the charge-density form $C$ is a differential $n$-form, and the charge-current vector $J$ is simply a vector field on $X$.

In this case, the electromagnetic field $F$ is a 2-form, related to the distribution of charges by the Maxwell equations:
$$ \d F = 0 \qquad , \qquad \partial F = J^* \ . $$
\end{example}

\bigskip

Let us consider a distribution of charged $p$-branes, represented by a charge-density $q$-form $C$ or, equivalently, by a charge-current $(p+1)$-vector $J$. Such a distribution of charges ``produces" an electromagnetic field, represented by a $(p+2)$-form $F$. 
By analogy with the particle case, we postulate that both fields are related by the following:
$$ \textbf{Maxwell Equations:} \qquad \boxed{ \phantom{\frac{2}{1}} \d F \,=\, 0 \quad , \quad \partial F\, = \, J^*  \ }$$ 
or equivalently
$$   \d F\, =\, 0 \quad , \quad \d (*F) \,=\, (-1)^p C  \ . $$

\smallskip

\begin{rem}
The second Maxwell equation implies an infinitesimal charge conservation law $\d C = 0$ (equivalently, $\partial J^* = 0 $).
\end{rem}\smallskip

\begin{rem} The operators $\d$ and $\partial$ are, essentially, the only first-order natural linear differential operators between differential forms. Therefore, in a certain sense, Maxwell equations are the only possible first-order equations that may arise.
 \end{rem}

\smallskip

\noindent\textbf{Variational principles.} In a similar vein to what is done for charged particles ($p=0$), the Lorentz force law and the Maxwell equations may be derived from variational principles, as follows.

Let us write $F=\text{d}A$ where $A$ is a $(p+1)$-form on spacetime, called the \textbf{electromagnetic potential}. 
 For each trajectory $S$ of a $p$-brane  with tension $\mathfrak{t}$ and electric charge $\mathfrak{q}$, consider the action:
$$\mathcal{A}(S):=\,-\mathfrak{t}\int_S\text{d}S+(-1)^p\mathfrak{q}\int_SA\ .$$
Extremals of this action are precisely the trajectories that satisfy the Lorentz force law.

On the other hand, let us fix a closed $q$-form $C$ on the spacetime $X$. For any $(p+1)$-form $A$, consider the action:
$$\mathcal{A}(A):=\, \int_X\frac{1}{2}F\wedge *F-\int_XA\wedge C  \ ,$$ where $F:=\text{d}A$. 
The Euler-Lagrange equations for this action amount to the Maxwell equation $\,\partial F=J^*$, where $i_J \d X = C$.

\section{Fluid of charged $p$-branes}

Now we shall extend in a natural way the notion of impulse of a $p$-brane and the force law to the case of a fluid of charged $p$-branes. Via \ref{energytensor}, we shall define the electromagnetic energy tensor associated to the electromagnetic field strength $F$, which is necessary to formulate the Einstein equation.


\subsection{Impulse form and force law for a fluid}

The following lemma is easy to check (v.gr. \cite{LovelockAlberto}, Lemma 2.4):

\begin{lem}\label{FormasYTensores}
 The following linear map is an isomorphism:
 $$ TX \otimes TX \xrightarrow{\  \ } \Lambda^nX \otimes TX \quad , \quad T^2 \mapsto C_1^1 ( \d X \otimes T^2) $$ where $C_1^1$ denotes the contraction between the first covariant and first contravariant indices.
 
 Moreover, if $T^2$ is a 2-contravariant tensor on $X$ and $\Pi_n := C_1^1( \d X \otimes T^2)$ is the corresponding vector-valued $n$-form, then:
 $$ \d_\nabla \Pi_n \,=\, \d X \otimes \mathrm{div}\, T^2 \ .$$
\end{lem}

If $T^2$ is a $2$-contravariant tensor, we write $T_2$ for the $2$-covariant tensor metrically equivalent to it.

\medskip

\smallskip

\begin{defn}
The mass-energy-momentum distribution of a fluid of charged $p$-branes is represented by a differential $n$-form $\Pi_n$ with values on $TX$, that we call \textbf{impulse form} 
of the fluid. 

The 2-covariant tensor $T_2$ corresponding to $\Pi_n$ via the isomorphism of Lemma \ref{FormasYTensores} is called the \textbf{stress-energy tensor} of the fluid.
\end{defn}

\smallskip

The interpretation of the impulse $n$-form $\Pi_n$ is the following: assume the ambient manifold $X$ is the Minkowski spacetime and let $H$ be an oriented hypersurface. 

If $S$ is the trajectory of a $p$-brane transversally crossing the hypersurface $H$, and $\Pi_S$ is the vector-valued impulse $p$-form of the $p$-brane, then the vector
 $\,\int_{S\cap H}\Pi_S\,$ is said to be the {\it total impulse} of the particle in the hypersurface $H$ (see Example \ref{PrimerMinkowski}).
 
Now, the vector
$$\int_H\Pi_n$$
is understood as {\it the sum of total impulses of all the charged $p$-branes transversally crossing the hypersurface $H$.}\bigskip

Let us consider a fluid of charged $p$-branes, with impulse form $\Pi_n$ and 
charge-current $(p+1)$-vector $J$.

In absence of external forces, the variation of the fluid impulse should be null: $\d_\nabla\Pi_n=0$. If an electromagnetic field  is present we postulate, by analogy with the case of a single $p$-brane,  that the movement of the fluid   satisfies the Force Law:
$$ \d_\nabla \Pi_n \,=\, \d X \otimes \widehat F(J)  \ . $$

In virtue of Lemma \ref{FormasYTensores}, this equation is equivalent to $\mathrm{div}\, T^2 \,=\,  \widehat F(J) $, or to $\mathrm{div}\, T_2\,=\, i_JF\ $. 
Combining it with the  Maxwell equation $\,\partial F=J^*$, we obtain another equivalent formulation:
$$\boxed{\phantom{\frac{1}{2}} \mathrm{div}\, T_2\,=\, i_{\partial F}F\quad}\ .$$

\subsection{Example: Dust of charged $p$-branes}

Let us consider a fluid of charged $p$-branes without pressure and where all the $p$-branes in the fluid have the same tension 
$\mathfrak{t}$ and the same electric charge $\mathfrak{q}$. The general idea is that   each $p$-brane has
approximately the same velocity as the surrounding ones; hence, we give the following definition:

 A \textbf{dust} of $p$-branes is described by an integrable distribution on $X$ of rank  $p+1$, for which each integrable 
submanifold represents the mean trajectory of an infinitesimal portion of $p$-branes.

Let $(D_0, \ldots , D_p)$ be an orthonormal basis $(+, - , \ldots , -)$ of the distribution. Such a basis defines an orientation on 
each integral submanifold. By analogy with the case of a single $p$-brane, the charge-current $(p+1)$-vector of the dust is defined as:
$$ J:= \,\rho_e D_0 \wedge \ldots  \wedge D_p \ , $$ for some \textbf{charge density} function $\rho_e$.

The contravariant stress-energy tensor of the dust is defined by the formula:
$$ T^2 := \rho_m \sum_{j=0}^p \delta_j D_j \otimes D_j $$ where $\delta_0 = 1$, $\delta_{j\neq0} = -1$, and the function $\rho_m := 
(\mathfrak{t}/\mathfrak{q}) \rho_e$ is called the \textbf{mass density} (that is, on each trajectory $S$ of the dust we consider the dual metric $(g_{|S})^* $
multiplied by the function $\rho_m$).

According \ref{FormasYTensores}, the corresponding impulse form is:
$$ \Pi_n \,=\, \rho_m \sum_{j=0}^p (i_{D_j} \d X) \otimes \delta_j D_j \ . $$ 

\smallskip

\begin{prop}
 If the charge conservation law $\d C = 0$ holds, then the impulse form of a charged dust satisfies:
$$ \d_\nabla \Pi_n \,=\, \rho_m \, \d X \otimes \mathrm{tr}\, \Phi $$ or, equivalently,
$$ \mathrm{div}\, T^2 \,=\, \rho_m \, \mathrm{tr}\, \Phi \ , $$ where $\Phi$ is the second fundamental form of the trajectories of the dust.
\end{prop}

\begin{proof} Let us complete the orthonormal basis $(D_0,\dots, D_p)$ of the distribution up to an oriented orthonormal basis $(D_0,\dots, D_n)$ of tangent fields on $X$, and let $(\theta_0,\dots,\theta_n)$ be the corresponding dual basis of $1$-forms. 

The charge-density $q$-form is:
$$C\,=\, i_J\text{d}X\,=\, i_J(\theta_0\wedge\dots\wedge\theta_n)\,=\, \rho_e\,\theta_{p+1}\wedge\dots\wedge\theta_n\ .$$ 

We have:
\begin{align*}
\Pi_n\,&=\, \rho_m\sum_{j=0}^p(i_{D_j}\text{d} X)\otimes \delta_jD_j\,=\, \rho_m\sum_{j=0}^pi_{D_j}(\theta_0\wedge\dots\wedge\theta_p)\wedge(\theta_{p+1}\wedge\dots\wedge\theta_n)\otimes \delta_jD_j \\
&=\, \left(\sum_{j=0}^pi_{D_j}(\theta_0\wedge\dots\wedge\theta_p)\otimes \delta_jD_j\right)\wedge(\rho_m\,\theta_{p+1}\wedge\dots\wedge\theta_n)\ .
\end{align*}
and therefore,
\begin{align*}
\text{d}_\nabla\Pi_n\,&=\,\left(\text{d}_\nabla\sum_{j=0}^pi_{D_j}(\theta_0\wedge\dots\wedge\theta_p)\otimes \delta_jD_j\right)\wedge(\rho_m\,\theta_{p+1}\wedge\dots\wedge\theta_n) \\
&\ +\ (-1)^p\left(\sum_{j=0}^pi_{D_j}(\theta_0\wedge\dots\wedge\theta_p)\otimes \delta_jD_j\right)\wedge\text{d}(\rho_m\,\theta_{p+1}\wedge\dots\wedge\theta_n)\ .
\end{align*}

The second addend is null because
 $\,\text{d}(\rho_m\,\theta_{p+1}\wedge\dots\wedge\theta_n)=\text{d}(\frac{\mathfrak{t}}{\mathfrak{q}}C)=0$. With respect to the first one, the term which is differentiated has the same expression than the impulse form of each integral submanifold (considered as the trajectory $S$ of a $p$-brane of tension $1$). 
 Applying Proposition \ref{DiferencialImpulso}, we obtain:
\begin{align*}
\text{d}_\nabla\Pi_n\,&=\,\left(\text{d}_\nabla\sum_{j=0}^pi_{D_j}(\theta_0\wedge\dots\wedge\theta_p)\otimes \delta_jD_j\right)\wedge(\rho_m\,\theta_{p+1}\wedge\dots\wedge\theta_n) \\
&=\, \left(\theta_0\wedge\dots\wedge\theta_p\otimes\text{tr}\,\Phi\right)\wedge(\rho_m\,\theta_{p+1}\wedge\dots\wedge\theta_n)\,=\,\rho_m\text{d}X\otimes\text{tr}\,\Phi\quad .\end{align*} 
\end{proof}

\smallskip

As a consequence, if the electromagnetic field ${F}$ satisfies the Maxwell equations (so, in particular, 
the charge conservation law holds), then, for a dust, the force law $\,\mathrm{div}\, T^2 = \widehat{F} (J)\,$ is equivalent to:
$$\rho_m \, \mathrm{tr}\, \Phi \,=\, \widehat{F} (J)  \ . $$


\subsection{Electromagnetic energy tensor}

\begin{defn} Let $F$ be an electromagnetic field. Its \textbf{electromagnetic energy tensor} is the 2-covariant tensor $T_{\text{elm}}$ associated to the $(p+2)$-differential form $F$ according to Definition \ref{energytensor}.
\end{defn}

Let $F$ be the electromagnetic field produced by a fluid of charged $p$-branes with stress-energy tensor $T_{\text{m}}$. The Lorentz force law $\,\text{div}\, T_\text{m}=i_{\partial F}F\,$ and Proposition \ref{divergence} produce an infinitesimal conservation law:
$$\mathrm{div}\,(T_\text{m}+T_\text{elm})\,=\, i_{\partial F}F+(-i_{\partial F}F)\,=\, 0\ .$$

Indeed, this property is the main motivation for the definition of the electromagnetic energy tensor (see the Introduction).\smallskip

Finally, as in the particle case, we postulate that the electromagnetic energy has a gravitational effect through the Einstein equation. 

To sum up, a \textbf{fluid of charged $p$-branes} is described by four tensor fields on spacetime: A stress-energy tensor $T_m$ and a charge-current $(p+1)$-vector $J$ representing the distributions of mass and charge, and a differential $(p+2)$-form $F$ and its energy tensor $T_{elm}$, representing the electromagnetic field and its electromagnetic energy.

They are related by the following equations:
\begin{align*}
 & \textbf{Maxwell equations:} \qquad  \phantom{\frac{2}{1}} \d F \,=\, 0 \quad , \quad \partial F\, = \, J^*  \  \\
 & \textbf{Einstein equation:} \hskip 1.5cm  Ricci(g)-\frac{r(g)}{2}\, g \, =\, T_m+T_\text{elm}\quad \ .
\end{align*}




\section*{Acknowledgements}

The authors acknowledge the referee for pointing out reference \cite{Henneaux}.


\begin{thebibliography}{60}

\bibitem{Anderson} {\sc Anderson, I.M.}: On the characterization of energy-momentum tensors, 
{\it Gen. Relativity Gravitation}, \textbf{10}, 461--466, (1979).

\bibitem{Jaime} {\sc Castrill\'{o}n, M., Mu\~{n}oz, J.}:  Gauge-invariant characterization of Yang-Mills-Higgs lagrangians,
{\it Ann. Henri Poincar\'{e}}, \textbf{8}, 203--217 (2007).

\bibitem{Epstein} {\sc Epstein,, D.B.A.}: Natural tensors on riemannian manifolds, {\it J. Differential Geom.} \textbf{10}, 631-645 (1975).

\bibitem{Gilkey} {\sc Gilkey, P. B.}: Curvature and the eigenvalues of the Laplacian for elliptic complexes, {\it Adv. Math.} \textbf{10}, 344-382 (1973).

\bibitem{Henneaux} {\sc Henneaux, M., Teitelboim, C.}: $p$-Form Electrodynamics, {\it Found. Phys.}, \textbf{16}, 7, 593-617, (1986).

\bibitem{Kerrighan1} {\sc Kerrighan, D.B.}: Arbitrary Tensor Concomitants of a Bivector  and a Metric in a Space-Time Manifold, {\it Gen. Relativity Gravitation} \textbf{13}, 19-27, (1981).

\bibitem{Kerrighan2} {\sc Kerrighan, D.B.}: On the uniqueness of the energy-momentum tensor for electromagnetism, {\it J. Math. Phys.} \textbf{23} (10), 1979-1980 (1982).

\bibitem{Kolar} {\sc Kol\'{a}r, I., Michor, P.W., Slov\'{a}k, J.}:  Natural Operations in Differential Geometry, Springer-Verlag, Berlin, (1993).

\bibitem{Lovelock} {\sc Lovelock, D.}: The uniqueness of the Einstein-Maxwell field equations, {\it Gen. Relativity Gravitation},
\textbf{5}, 399--408, (1974).

\bibitem{Lovelock2} {\sc Lovelock, D.}: Bivector field theories, divergence-free vectors and the Einstein-Maxwell filed equations, 
{\it J. Math. Phys.}, 
\textbf{18}, 1491--1498, (1977).


\bibitem{LovelockAlberto} {\sc Navarro, A., Navarro, J.:} Lovelock's theorem revisited, {\it J. Geom. Phys.}, \textbf{61}, 1950--1956 (2011).

\bibitem{Einstein} {\sc Navarro, J., Sancho, J. B.}: On the naturalness of Einstein equation, {\it J. Geom. Phys.}, \textbf{57}, 1007--1014 (2008).

\bibitem{Superenergy} {\sc Senovilla, J.M.M.}: Superenergy tensors, {\it Classical Quantum Gravity}, \textbf{17}, 2799-2841, (2000).

\bibitem{Peetre} {\sc Slov\'{a}k, J.}: Peetre theorem for nonlinear operators, {\it Ann. Global Anal. Geom.}, \textbf{6}, 273--283, 
(1988).

\bibitem{SlovakJournal} {\sc Slov\'{a}k, J.}: On invariant operations on a manifold with connection or metric, {\it  J. Differential Geom.}, \textbf{36}, 633--650, (1992).

\bibitem{SlovakSemi} {\sc Slov\'{a}k, J.}: On invariant operations on pseudo-riemannian manifolds, {\it Comment. Math. Univ.
Carolin.}, \textbf{33}, 2, 269--276, (1992).


\bibitem{Stredder} {\sc Stredder, P.}: Natural differential operators on riemannian manifolds and
representations of the orthogonal and special orthogonal groups, {\it  J. Differential Geom.},  \textbf{10}, 647--660, (1975).


\bibitem{Taylor} {\sc Taylor, M.E.}: Partial Differential Equations III, Nonlinear equations, Springer (1997).




\end{thebibliography}
\end{document}